\documentclass[journal,10pt]{IEEEtran}
\usepackage{algorithmic}
\usepackage[cmex10]{amsmath}
\usepackage{graphicx}
\usepackage{amssymb}
\usepackage{cite}
\usepackage{epsfig}
\usepackage{amsthm}
\usepackage{color}
\usepackage[tight,footnotesize]{subfigure}
\usepackage{flushend}
\flushend

\usepackage{pgfplotstable}
\usepackage{booktabs}
\usepackage{colortbl}

\theoremstyle{theorem}
\newtheorem{theorem}{Theorem}
\newtheorem{lemma}[theorem]{Lemma}
\newtheorem{corollary}[theorem]{Corollary}

\theoremstyle{definition}
\newtheorem{definition}{Definition}

\theoremstyle{remark}
\newtheorem{remark}{Remark}

\usetikzlibrary{arrows,decorations.markings}
\usepackage{tikz}

\ifCLASSINFOpdf
\else
\fi
\begin{document}
%
\title{Construction of Subexponential-Size Optical Priority Queues with Switches and Fiber Delay Lines}
%
%
%

\author{Bin~Tang,~\IEEEmembership{Member,~IEEE,}
        Xiaoliang~Wang,~\IEEEmembership{Member,~IEEE,}
        Cam-Tu Nguyen,~\IEEEmembership{Member,~IEEE,}
        Baoliu~Ye,~\IEEEmembership{Member,~IEEE,}
        and~Sanglu~Lu,~\IEEEmembership{Member,~IEEE}
\thanks{A preliminary version of the paper has appeared in IEEE International Symposium on Information Theory, Barcelona, Spain, July 10-15, 2016~\cite{tang2016constructing}. }
\thanks{B. Tang, X. Wang, C.-T. Nguyen, B. Ye, and S. Lu are with the National Key Laboratory for Novel Software Technology, Nanjing University, Nanjing 210023, China (e-mail: \{tb, waxili, ncamtu, yebl, sanglu\}@nju.edu.cn).}}

%
%

\markboth{IEEE/ACM Transactions on Networking, manuscript}%
{Tang \MakeLowercase{\textit{et al.}}: Construction of Subexponential-Size Optical Priority Queues with Switches and Fiber Delay Lines}
%



\maketitle

\begin{abstract}

All-optical switching has been considered as a natural choice to keep pace with growing fiber link capacity. One key research issue of all-optical switching is the design of optical buffers for packet contention resolution. One of the most general buffering schemes is optical priority queue, where every packet is associated with a unique priority upon its arrival and departs the queue in order of priority, and the packet with the lowest priority is always dropped when a new packet arrives but the buffer is full. In this paper, we focus on the feedback construction of an optical priority queue with a single $\boldsymbol{(M+2)\times (M+2)}$ optical crossbar Switch and $\boldsymbol{M}$ fiber Delay Lines (SDL) connecting $\boldsymbol{M}$ inputs and $\boldsymbol{M}$ outputs of the switch. We propose a novel construction of an optical priority queue with buffer $\boldsymbol{2^{\Theta(\sqrt{M})}}$, which improves substantially over all previous constructions that only have buffers of $\boldsymbol{O(M^c)}$ size for constant integer $\boldsymbol{c}$. The key ideas behind our construction include (i) the use of first in first out multiplexers, which admit efficient SDL constructions, for feeding back packets to the switch instead of fiber delay lines, and (ii) the use of a routing policy that is similar to self-routing, where each packet entering the switch is routed to some multiplexer mainly determined by the current ranking of its priority.

\end{abstract}


%
\IEEEpeerreviewmaketitle

\section{Introduction}
\label{sec:introduction}
%

\IEEEPARstart{A}{ll-optical} packet switching is very attractive for making a good use of the enormous bandwidth of optical networks, since it eliminates the complicated and quite expensive optical-electrical-optical conversions. One main issue for implementing all-optical packet switching is the construction of optical buffers for conflict resolutions among packets competing for the same resources. As optical-RAM is not available yet, a common approach for constructing optical buffers is to use a combination of bufferless optical crossbar Switches and fiber Delay Lines (SDLs), where fiber delay lines (FDLs) act as storage devices for optical packets~\cite{karol1993shared,chlamtac1996cord,cruz1996cod,hunter1998slob}. However, unlike the traditional electronic memories with random access, one packet entering an FDL must propagate for a fixed amount of time and cannot be retrieved anytime earlier. Such inflexibility makes the design of SDL-based optical buffers with the same throughput and delay performance as its electronic counterpart quite challenging. In the past one decade and a half, great efforts have been made on constructing various kinds of optical buffers, such as first in first out (FIFO) multiplexers~\cite{chang2004recursive,chang2006using,chou2006necessary,cheng2007constructions,chen2007feedforward,cheng2008constructions,cheng2017greedy}, FIFO queues~\cite{chang2006constructions,li2011mux,huang2007recursive,cheng2013necessary}, last in first out (LIFO) queues~\cite{huang2007recursive,small2007modular,wang2011efficient}, priority queues~\cite{sarwate2006exact,chiu2007simple,chiu2007using,kogan2007optimal,cheng2011constructions,datta2017construction}, and shared queues~\cite{wang2009construction,wang2012constructing}, etc.

In this paper, we focus on the design of optical \emph{priority queues} with SDLs. A priority queue contains an arrival link, a departure link, and a loss link. Each packet is associated with a unique priority upon its arrival. When a departure request is raised by a controller, the packet with the highest priority is sent out from the departure link. If a new packet arrives but the buffer of the priority queue is full, then the packet with the lowest priority is dropped via the loss link. Priority queue is one of the most general buffering schemes, as the priority of each packet can be assigned arbitrarily. In particular, both FIFO queues and LIFO queues can be viewed as priority queues where the arrival time of a packet is used as its priority.

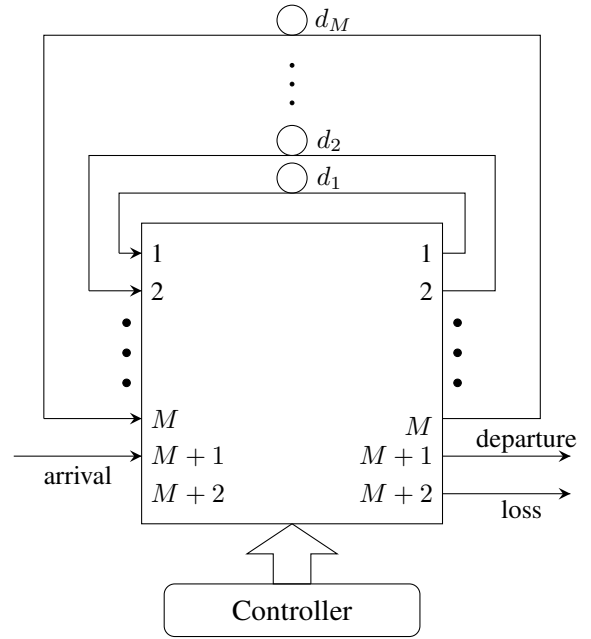
\begin{figure}[!tb]
    \centering
       \begin{tikzpicture}[scale=1]
    	\tikzset{myptr/.style={decoration={markings,mark=at position 1 with %
   			{\arrow[scale=1.4,>=stealth]{>}}},postaction={decorate}}}
        \draw (3.5,3.5) rectangle (7.5,7.5);

        \draw[rounded corners] (3.8,2) rectangle (7.2,2.7);
        \node[scale=1.1] at (5.5,2.35){Controller};

        \draw (5.5,3.5)--(6.1,3.1)--(5.75,3.1)--(5.75,2.7);
        \draw (5.5,3.5)--(4.9,3.1)--(5.25,3.1)--(5.25,2.7)--(5.75,2.7);

        \node[right] at (3.5,7.1){1};
        \node[right] at (3.5,6.6){2};
        \node[right] at (3.5,4.9){$M$};
        \node[right] at (3.5,4.4){$M+1$};
        \draw[myptr](1.8,4.4)--node [pos=0.5,below,sloped]{arrival}(3.5,4.4);
        \node[right] at (3.5,3.9){$M+2$};
        \draw [fill] (3.3,6.175) circle [radius=0.05];
        \draw [fill] (3.3,5.775) circle [radius=0.05];
        \draw [fill] (3.3,5.375) circle [radius=0.05];

        \node[left] at (7.5,7.1){1};
        \node[left] at (7.5,6.6){2};
        \node[left] at (7.5,4.8){$M$};
        \node[left] at (7.5,4.4){$M+1$};
        \draw[myptr](7.5,4.4)--(9.2,4.4);
        \node [left] at (9.4,4.6){departure};
        \node[left] at (7.5,3.9){$M+2$};
        \draw[myptr](7.5,3.9)--(9.2,3.9);
        \node [left] at(8.95,3.7){loss};

        \draw [fill] (7.7,6.175) circle [radius=0.05];
        \draw [fill] (7.7,5.775) circle [radius=0.05];
        \draw [fill] (7.7,5.375) circle [radius=0.05];

        \draw (7.5,7.1)--(7.8,7.1)--(7.8,7.9)--(3.2,7.9)--(3.2,7.1);
        \draw[myptr] (3.2,7.1)--(3.5,7.1);
        \draw (5.5,8.1) circle [radius=0.2];
        \node at(6,8.1){$d_1$};

        \draw (7.5,6.6)--(8.2,6.6)--(8.2,8.4)--(2.8,8.4)--(2.8,6.6);
        \draw[myptr] (2.8,6.6)--(3.5,6.6);
        \draw (5.5,8.6) circle [radius=0.2];
        \node at(6,8.6){$d_2$};

        \draw (7.5,4.9)--(8.8,4.9)--(8.8,10)--(2.2,10)--(2.2,4.9);
        \draw[myptr] (2.2,4.9)--(3.5,4.9);
        \draw (5.5,10.2) circle [radius=0.2];
        \node at(6.05,10.2){$d_M$};

        \draw [fill] (5.5,9.1) circle [radius=0.025];
        \draw [fill] (5.5,9.35) circle [radius=0.025];
        \draw [fill] (5.5,9.6) circle [radius=0.025];
    \end{tikzpicture}
    \caption{Construction of an optical priority queue with an $(M+2)\times (M+2)$ optical crossbar switch and $M$ fiber delay lines with delays $d_1,d_2,\ldots,d_M$.}
    \label{fig:general}
\end{figure}

Following previous works~\cite{sarwate2006exact,chiu2007simple,chiu2007using,datta2017construction}, we consider the construction of an optical priority queue using a feedback system as illustrated in Fig.~\ref{fig:general}. This system consists of an $(M+2)\times (M+2)$ optical crossbar switch, which has one distinguished input for external packet arriving, one distinguished output for packet departure, one distinguished output for packet loss, and $M$ FDLs with delays $d_1, d_2,\ldots,d_M$ connecting the other inputs and outputs in pairs. The issue is to choose proper delays $d_1, d_2,\ldots, d_M$ as well as the routing policy performed by the switch, such that the switching system can exactly emulate a priority queue.

All the arrival time and priorities of packets and the packet departure requests can be arbitrary, making the optical priority queue highly dynamic. This leads to the design of delays of FDLs and the routing policy in a coupled way very difficult. In particular, there are two basic necessary conditions for the routing policy:
\begin{itemize}
  \item \emph{Delay} condition: a packet with the $i$-th highest priority cannot be switched into an FDL with delay higher than $i$.
  \item \emph{Collision-free} condition, i.e., for any time and any FDL, there must be at most one packet entering the FDL.
%
\end{itemize}
%
%
Based on these conditions, Sarwate and Anatharam~\cite{sarwate2006exact} showed that the buffer size is upper bounded by $2^M+1$.
To accommodate the conditions, they introduced a routing policy based on sorting the priorities of the packets entering the switch. Proper delays were further assigned to the FDLs, which leads to the first construction of optical priority queue with buffer $\Theta(M^2)$~\cite{sarwate2006exact}. This sorting-based routing policy plays a vital role in all the subsequent constructions of optical priority queues, including the ones by Chiu \emph{et al.} in \cite{chiu2007simple} and \cite{chiu2007using} whose buffer sizes are $\Theta (M^2)$ and $\Theta (M^3)$, respectively, and the recursive construction by Datta~\cite{datta2017construction},
which can achieve a buffer size of $\Theta(M^c)$ for any positive integer $c$.
\footnote{Datta's work~\cite{datta2017construction} and our preliminary version of this work~\cite{tang2016constructing} firstly appeared at almost the same time.} However, all these buffer sizes achieved are polynomial in $M$, which are far away from the exponential upper bound $2^M+1$.

In this paper, we make a great step towards closing the above gap by presenting a novel construction of an optical priority queue with buffer $2^{\Theta(\sqrt{M})}$. To the best of our knowledge, this is the first construction of an optical priority queue whose buffer size goes beyond polynomials of the number of FDLs $M$.
The key ideas behind our construction include two aspects.

\begin{itemize}
  \item As illustrated in Fig.~\ref{fig:basicconstruction}, we use (FIFO) multiplexers for feeding back optical packets to the switch instead of the direct use of FDLs. A multiplexer has multiple input links for packet arrivals, one output link for packet departure, and some other output links for packet loss. It allows multiple packets to arrive simultaneously, and at each time slot there is always a packet departing in the FIFO order whenever the multiplexer is nonempty. Although a multiplexer with $\tilde{B}$ buffer needs a crossbar switch and $O(\log \tilde{B})$ FDLs for construction~\cite{chang2004recursive}, the collision-free condition can be relaxed when replacing FDLs with multiplexers, since each multiplexer can accept the entrance of multiple packets simultaneously, which brings extra room for the design of routing policy. On the other hand, the use of multiplexers imposes an additional condition on the routing policy that buffer overflow cannot happen at any multiplexer. Nevertheless, we only need to guarantee that the number of packets buffered at a multiplexer cannot exceed the buffer size of the multiplexer, since the buffer space of a multiplexer is always used in a consecutive manner.

  \item We introduce a novel routing policy that is similar to self-routing~\cite{chang2004recursive}, where each packet entering the switch is routed to some multiplexer mainly determined by the current ranking of its priority according to a simple routing rule. Compared to the sorting-based routing policy used in all previous constructions, our routing policy also incurs a lower computation cost.
\end{itemize}

\begin{figure}[!tb]
    \centering
\begin{tikzpicture}[scale=0.9]
    	\tikzset{myptr/.style={decoration={markings,mark=at position 1 with
   			{\arrow[scale=1.4,>=stealth]{>}}},postaction={decorate}}}

        \draw (3.5,3.5) rectangle (7.5,7.5);
        \node[scale=1.5] at (5.5,5.5){Switch};

        \draw[rounded corners] (3.8,2) rectangle (7.2,2.7);
        \node[scale=1.1] at (5.5,2.35){Controller};

        \draw (5.5,3.5)--(6.1,3.1)--(5.75,3.1)--(5.75,2.7);
        \draw (5.5,3.5)--(4.9,3.1)--(5.25,3.1)--(5.25,2.7)--(5.75,2.7);

        \draw[myptr](1.6,4.1)--node [pos=0.5,below,sloped]{arrival}(3.5,4.1);

        \draw[myptr](7.5,4.1)--(9.5,4.1);
        \node [left] at(9.45,4.3){departure};

        \draw[myptr](7.5,3.7)--(9.5,3.7);
        \node [left] at(9.1,3.5){loss};

        \draw (4.5,7.7) rectangle (6.5,8.5);
        \node[scale=0.9,align=center] at (5.5,8.1){
        	multiplexer};
        \foreach \x in {1,4}
        {
        	\draw (7.5,7.6-0.2*\x)--(7.6+0.1*\x,7.6-0.2*\x)--(7.6+0.1*\x,7.6+0.2*\x);
        	\draw[myptr] (7.6+0.1*\x,7.6+0.2*\x)--(6.5,7.6+0.2*\x);
        }
    	\draw (4.5,8.1)--(3,8.1)--(3,7.1);
    	\draw[myptr] (3,7.1)--(3.5,7.1);
    \draw [fill] (7,8.1) circle [radius=0.015];
    \draw [fill] (7,8) circle [radius=0.015];
    \draw [fill] (7,8.2) circle [radius=0.015];

    \draw (4.5,8.7) rectangle (6.5,9.5);
    \node[scale=0.9,align=center] at (5.5,9.1){
    	multiplexer};
    \foreach \x in {1,4}
    {
    	\draw (7.5,6.8-0.2*\x)--(8.1+0.1*\x,6.8-0.2*\x)--(8.1+0.1*\x,8.6+0.2*\x);
    	\draw[myptr] (8.1+0.1*\x,8.6+0.2*\x)--(6.5,8.6+0.2*\x);
    }
	\draw (4.5,9.1)--(2.5,9.1)--(2.5,6.3);
	\draw[myptr] (2.5,6.3)--(3.5,6.3);
 \draw [fill] (7,9.1) circle [radius=0.015];
    \draw [fill] (7,9) circle [radius=0.015];
    \draw [fill] (7,9.2) circle [radius=0.015];
    \draw [fill] (3.3,5.8) circle [radius=0.025];
    \draw [fill] (3.3,5.6) circle [radius=0.025];
    \draw [fill] (3.3,5.4) circle [radius=0.025];

    \draw [fill] (7.7,5.8) circle [radius=0.025];
    \draw [fill] (7.7,5.6) circle [radius=0.025];
    \draw [fill] (7.7,5.4) circle [radius=0.025];

    \draw [fill] (5.5,9.7) circle [radius=0.025];
    \draw [fill] (5.5,9.9) circle [radius=0.025];
    \draw [fill] (5.5,10.1) circle [radius=0.025];

    \draw [fill] (8.8,7.6) circle [radius=0.025];
    \draw [fill] (9,7.6) circle [radius=0.025];
    \draw [fill] (9.2,7.6) circle [radius=0.025];

    \draw [fill] (2.1,7.6) circle [radius=0.025];
    \draw [fill] (1.9,7.6) circle [radius=0.025];
    \draw [fill] (1.7,7.6) circle [radius=0.025];

    \draw (4.5,10.3) rectangle (6.5,11.1);
    \node[scale=0.9,align=center] at (5.5,10.7){
    	multiplexer};
    \foreach \x in {1,4}
    {
    	\draw (7.5,5.4-0.2*\x)--(9.4+0.1*\x,5.4-0.2*\x)--(9.4+0.1*\x,10.2+0.2*\x);
    	\draw[myptr] (9.4+0.1*\x,10.2+0.2*\x)--(6.5,10.2+0.2*\x);
    }
	\draw (4.5,10.7)--(1.2,10.7)--(1.2,4.9);
	\draw[myptr] (1.2,4.9)--(3.5,4.9);
 \draw [fill] (7,10.6) circle [radius=0.015];
    \draw [fill] (7,10.8) circle [radius=0.015];
    \draw [fill] (7,10.7) circle [radius=0.015];
    \end{tikzpicture}

     \caption{Illustration of the multiplexer based construction of an optical priority queue. Here the loss links of multiplexers are omitted.}
      \label{fig:basicconstruction}
\end{figure}
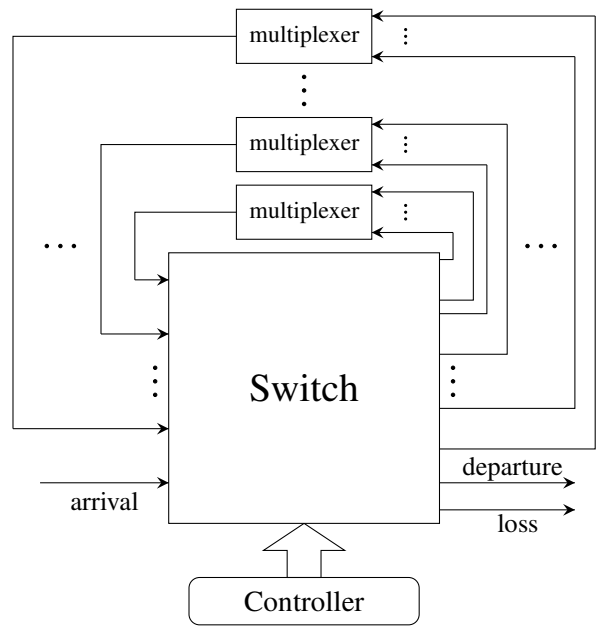

Specifically, we adopt 4-to-1 multiplexers and use them in groups each of which consists of three same 4-to-1 multiplexers. By using an exponential sequence for setting the buffer sizes of multiplexers and an appropriate routing rule, we can guarantee that neither packet collision nor buffer overflow could happen at each multiplexer. Based on these salient properties, we show that our construction emulates a priority queue exactly.  Although our construction uses multiple switches, we can combine all the switches into one, and finally have a construction of an optical priority queue with buffer $2^{\Theta(\sqrt{M})}$ using a single crossbar switch and $M$ fiber delay lines.

The remainder of this paper is organized as follows. In Sec.~\ref{sec:preliminaries}, we introduce the basic assumptions and definitions used throughout this paper. In Sec.~\ref{sec:construction}, we present a very efficient construction of optical priority queues while the proof is given in Sec.~\ref{sec:proof}. Sec.~\ref{sec:relatedwork} discusses about related work. Finally, Sec.~\ref{sec:conclusion} presents the concluding remarks.

\section{Preliminaries}
\label{sec:preliminaries}

In this section, we first introduce the basic assumptions and network elements adopted in this paper and then introduce the definition of priority queue.

\subsection{Assumptions and Basic Network Elements}
As in most work about the SDL-based optical queue designs~\cite{chang2004recursive,chang2006using,chou2006necessary,cheng2007constructions,chen2007feedforward,
cheng2008constructions,cheng2017greedy,chang2006constructions,li2011mux,huang2007recursive,cheng2013necessary,small2007modular,wang2011efficient, sarwate2006exact,chiu2007simple,chiu2007using,kogan2007optimal,cheng2011constructions,datta2017construction,wang2009construction,wang2012constructing},
we assume that the time of system is slotted and synchronized, and the packet size is fixed such that one packet can be transmitted over a link within one time slot. Since there is at most one packet in a link, we can use 0-1 variables to characterize the state of a link. We say that a link is in state 1 at time $t$ if there is a packet in the link at $t$, and the link is in state 0 at $t$ otherwise.

Switches and fiber delay lines are defined as follows.

\begin{definition}[\textbf{Switch}]
  An $n\times n$ (optical) crossbar switch is a \emph{memoryless} network element that has $n$ input links and $n$ output links, which can realize all the $n!$ permutations between its inputs and outputs. Specifically, for any $k$, $k\leq n$, packets coming from any $k$ input links will instantaneously go out from $k$ output links which are specified by a protocol performed by the switch. We will refer to $n$ as the size of the switch and the protocol as the \emph{routing policy} of the switch.
\end{definition}

\begin{definition}[\textbf{Fiber delay line, FDL}]
A fiber delay line with delay $d$ (a non-negative integer) is a network element that has one input link and one output link, through which $d$ time slots are required for a packet to traverse. Let $a(t)$ denote the state of the input link at time $t$. Then the state of the output link at $t$ is $a(t-d)$.
\end{definition}

When a packet is traversing through an FDL, it looks like that the packet is buffered in the FDL. Therefore, an FDL can be viewed as a memory device, but it is much more inflexible than traditional electronic memory since at most one packet can enter the FDL at one time slot and a packet entering the the FDL can only be retrieved after a fixed amount of time.

\subsection{Priority Queues}

Consider the network element shown in Fig.~\ref{fig:priorityQ}, which has an input link for packet arrival, one controller, and two output links, one for departing packets, and the other for loss packets. Every packet arriving at the network element is associated with a unique label, called \emph{priority}, which is used to indicate the expected departure order of this packet among all the buffered packets. Suppose there are $k$ packets at the beginning of time $t$, including the arriving packet if any, in the switching system. If a packet $i$ has the $j$-th highest priority among the $k$ packets, we say that $i$ has a tag of $j$ at time $t$, which is denoted by $\tau_i(t)=j$. Hence, a packet having a smaller tag has a higher priority than a packet having a larger tag at any time. However, the tag of a packet buffered in the system can change over time due to the arrival and departure of other packets.

 We use the following notations to describe the state of the network element at each time $t$.
\begin{itemize}
  \item Let $a(t)$, $d(t)$ and $l(t)$ denote the states of the input link, the departure link and the loss link at time $t$, respectively.
  \item Let $c(t)=1$ if the controller sends a departure request at time $t$ and $c(t)=0$ otherwise.
  \item Denote by $q(t)$ the number of packets buffered in the network element at time $t$.
\end{itemize}

A discrete-time priority queue can then be defined formally as follows.
\begin{definition}[\textbf{Priority Queue}]
\label{def:pq}
Starting empty at time 0, the network element shown in Fig.~\ref{fig:priorityQ} is called a priority queue with buffer $B$ if it satisfies all the following properties at each time $t>0$:
  \begin{enumerate}
\item[(P1)] \emph{Flow conservation}: arriving packets are either stored in the network element or transmitted through the departure link or the loss link, i.e.,
\begin{equation}
  q(t)=q(t-1)+a(t)-d(t)-l(t).
\end{equation}
\item[(P2)] \emph{Non-idling}: If there are packets buffered in the network element or there is an arriving packet, then there is a packet departing from the network element if and only if the controller sends a departure request, i.e.,
    \begin{equation}
      d(t)=
      \begin{cases}
        1 & \text{if }c(t)=1 \text{ and }q(t-1)+a(t)>0\\
        0 & \text{otherwise.}
      \end{cases}
    \end{equation}

\item[(P3)] \emph{Maximum buffer usage}: There is a packet dropped out from the loss link  if and only if there is no departure request, the buffer is full and there is an arriving packet, i.e.,
\begin{equation}
  l(t)=
  \begin{cases}
    1 & \text{if }c(t)=0, q(t-1)=B \text{ and }a(t)=1\\
    0 & \text{otherwise.}
  \end{cases}
\end{equation}

\item[(P4)] \emph{Priority departure}: If there is a departure packet $i$ at time $t$, then $i$ must have the highest priority among all the packets buffered in the network element and the arriving packet (if any) at time $t$, i.e.,
    \begin{equation}
      \tau_i(t)=1.
    \end{equation}

\item[(P5)] \emph{Priority loss}: If there is a loss packet $i$ at time $t$, then $i$ much have the lowest priority among all the $B$ packets buffered in the network element and the arriving packet at time $t$, i.e.,
    \begin{equation}
      \tau_i(t)=B+1.
    \end{equation}
\end{enumerate}
\end{definition}

If a priority queue is constructed with optical crossbar switches and FDLs, we say that it is an \emph{optical priority queue}. In this paper, we focus on the construction of optical priority queues with a single optical crossbar switch and $M$ FDLs as shown in Fig.~\ref{fig:general}. The efficiency of a construction is evaluated by the buffer size of the constructed optical priority queue in terms of $M$.

\begin{figure}
   \centering
\begin{tikzpicture}[scale=1]
    	\tikzset{myptr/.style={decoration={markings,mark=at position 1 with {\arrow[scale=1.0,>=stealth]{>}}},postaction={decorate}}}

    \draw (4,4) rectangle (8,5.6);
    \node[scale=1.1] at(6,4.8){priority queue buffer};

    \draw[myptr](2.5,4.8)--node [pos=0.5,above,sloped]{arrival}(4,4.8);

    \draw[myptr](8,5.1)--(9.5,5.1);
    \node[scale=1,right] at(8.1,5.35) {departure};

    \draw[myptr](8,4.5)--(9.5,4.5);
    \node[scale=1,right] at(8.4,4.3) {loss};

    \draw(4,5.7)--(4,6);
    \draw[myptr](5.5,5.85)--(4,5.85);
    \node at(6,5.85) {$\textbf{\textit{B}}$};
    \draw[myptr](6.5,5.85)--(8,5.85);
    \draw(8,5.7)--(8,6);

    \draw[rounded corners] (4.6,3) rectangle (7.4,3.5);
    \node[scale=0.9] at (6,3.25){Controller};

    \draw (6,4)--(6.4,3.7)--(6.15,3.7)--(6.15,3.5);
    \draw (6,4)--(5.6,3.7)--(5.85,3.7)--(5.85,3.5);

    \end{tikzpicture}
    \caption{A priority queue with $B$ buffer.}
    \label{fig:priorityQ}
\end{figure}
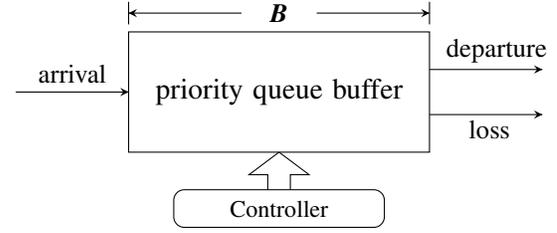

\subsection{Multiplexers}

Our construction of optical priority queue will use FIFO multiplexers as intermediate building blocks. For the sake of completeness, we give a formal definition of multiplexers.

\begin{definition}[\textbf{Multiplexer}]
  An $n$-to-1 (FIFO) multiplexer with buffer $\tilde{B}$ is a network element with $n$ input links, one departure link, and $n-1$ output links for packet losses. Let $\tilde{a}_i(t)$, $i=1,2,\ldots,n$, be the state of the $i$-th input link, $\tilde{d}(t)$ be the state of the departure link and $\tilde{l}_i(t)$, $i=1,2,\ldots,n-1$, be the state of the $i$-th loss link, and $\tilde{q}(t)$ be the number of packets buffered at the multiplexer at time $t$. The $n$-to-1 multiplexer with buffer $\tilde{B}$ satisfies the following four properties.
  \begin{itemize}
    \item [(M1)] Flow conservation: arriving packets from the $n$ input
      links are either stored in the buffer or transmitted through the $n$ output links, i.e.,
    \begin{equation}
      \tilde{q}(t)=\tilde{q}(t-1)+\sum_{i=1}^n\tilde{a}_i(t)-\tilde{d}(t)-\sum_{i=1}^{n-1} \tilde{l}_i(t).
    \end{equation}

    \item [(M2)] Non-idling: there is always a departing packet if there are packets in the buffer or there are arriving packets, i.e.,
    \begin{equation}
      \tilde{d}(t)=\begin{cases}
        1 & \textrm{if } \tilde{q}(t-1)+\sum_{i=1}^n \tilde{a}_i(t)>0\\
        0 & \textrm{otherwise.}
      \end{cases}
    \end{equation}

    \item [(M3)] Maximum buffer usage: arriving packets are lost only when the buffer is full, i.e., for $i=1,\ldots,n-1$,
    \begin{equation}
      \tilde{l}_i(t)=\begin{cases}
        1 & \textrm{if }\tilde{q}(t-1)+\sum_{i=1}^n\tilde{a}_i(t)\geq \tilde{B}+i+1\\
        0 & \textrm{otherwise.}
      \end{cases}
    \end{equation}

    \item [(M4)] FIFO: packets depart in the FIFO order.

  \end{itemize}
\end{definition}

\begin{figure}[tb]
    \centering
       \begin{tikzpicture}[>=stealth]
    	\tikzset{myptr/.style={decoration={markings,mark=at position 1 with {\arrow[scale=1.0,>=stealth]{>}}},postaction={decorate}}}

    \draw (4,4) rectangle (8,6);
    \node[scale=1.2] at(6,5.2){4-to-1 multiplexer};
 \node[scale=1.2] at(6,4.7){buffer};
    \draw[myptr](3,4.35)--(4,4.35);
    \draw[myptr](3,4.75)--(4,4.75);
    \draw[myptr](3,5.15)--(4,5.15);
    \draw[myptr](3,5.55)--(4,5.55);
    \draw[decorate, decoration={brace, mirror}](2.8,5.6) -- (2.8,4.3);
    \node[left] at(2.75,5){arrival};

    \draw[myptr](8,5.6)--(9,5.6);
    \node[scale=1,right] at(9,5.6) {departure};

    \draw[myptr](8,5.15)--(9,5.15);
    \draw[myptr](8,4.75)--(9,4.75);
    \draw[myptr](8,4.35)--(9,4.35);
    \draw[decorate, decoration={brace, mirror}](9.15,4.3) -- (9.15,5.2);
    \node[right] at(9.2,4.7){loss};

    \draw(4,6.2)--(4,6.5);
    \draw[myptr](5.5,6.35)--(4,6.35);
    \node at(6,6.35) {$\boldsymbol{\tilde{B}}$};
    \draw[myptr](6.5,6.35)--(8,6.35);
    \draw(8,6.2)--(8,6.5);

    \end{tikzpicture}
    \caption{A 4-to-1 multiplexer with $\tilde{B}$ buffer.}
    \label{fig:4to1}
\end{figure}
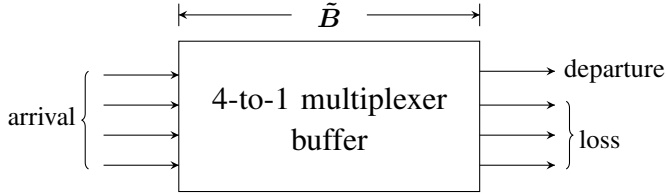

See Fig.~\ref{fig:4to1} for an illustration of a 4-to-1 multiplexer with buffer $\tilde{B}$. As also mentioned in Sec.~\ref{sec:introduction}, a multiplexer with buffer $\tilde{B}$ is much more flexible than an FDL with delay $\tilde{B}$. Specifically,
\begin{itemize}
  \item A multiplexer has multiple inputs, which brings extra room for the design of routing policy as the collision-free condition is easier to satisfy.

  \item The buffer of a multiplexer is always used in a consecutive manner, so it can be fully utilized, and as long as the number of packets buffered does not exceed the buffer size, there would never be any buffer overflow. On the other hand, it is very hard to fully use an FDL viewed as a buffer. See Fig.~\ref{fig:bufferstate} for an illustration.
\end{itemize}

\begin{figure}
\centering
    \begin{tikzpicture}
    [ L1Node/.style={rectangle,draw, minimum size=7mm}]

       \node[L1Node,fill=gray!50]at(5.6,2.5){};
       \node[L1Node,fill=gray!50]at(4.9,2.5){};
       \node[L1Node,fill=gray!50]at(4.2,2.5){};
       \node[L1Node]at(3.5,2.5){};
       \node[L1Node]at(2.8,2.5){};
       \node[L1Node]at(2.1,2.5){};
       \node[L1Node]at(1.4,2.5){};
       \draw (1.4,2.85)--(0,2.85);
       \draw (1.4,2.15)--(0,2.15);
       \node[right] at(1.9,1.8){(a) multiplexer};

       \node[L1Node]at(5.6,0.8){};
       \node[L1Node,fill=gray!50]at(4.9,0.8){};
       \node[L1Node]at(4.2,0.8){};
       \node[L1Node,fill=gray!50]at(3.5,0.8){};
       \node[L1Node,fill=gray!50]at(2.8,0.8){};
       \node[L1Node]at(2.1,0.8){};
       \node[L1Node]at(1.4,0.8){};
       \draw (1.4,1.15)--(0,1.15);
       \draw (1.4,0.45)--(0,0.45);
       \node[right] at(2.4,0.1){(b) FDL};

       \node[left] at(0,2.5){in};
       \node[right] at(5.95,2.5){out};
          \node[left] at(0,0.8){in};
       \node[right] at(5.95,0.8){out};
    \end{tikzpicture}
    \caption{An illustration of buffer states of a multiplexer and an FDL where each slot corresponds to a packet size, and a gray slot represents a packet. }
    \label{fig:bufferstate}
\end{figure}
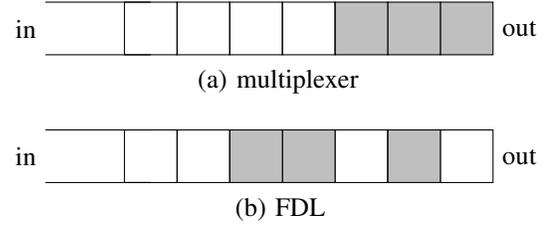
%
%

\section{Construction of Optical Priority Queues}
\label{sec:construction}
In this section, we present a very efficient construction of optical priority queues based on multiplexers, and analyze its construction cost in terms of SDLs.

To ease the presentation, we introduce some notations regarding sets of consecutive integers. Let $\Psi$ be a set of consecutive integers. Define $L(\Psi)$ and $U(\Psi)$ be the smallest integer and the largest integer in $\Psi$, respectively. That is, $\Psi=\{L(\Psi),L(\Psi)+1,\ldots,U(\Psi)\}$. For simplicity, we write $\Psi=\langle L(\Psi),U(\Psi)\rangle$.

In order to help understand our construction, we start by introducing the motivation behind our design idea.

\subsection{Motivation}
\label{sec:motivation}

Consider the construction of an optical priority queue using a feedback system as illustrated in Fig.~\ref{fig:general}, and suppose that there are $M=2\ell-1$ FDLs indexed by $1,2,\ldots,M$ for some positive integer $\ell$. One necessary condition for the design of delays of FDLs and the routing policy is that, a packet with the $i$-th highest priority cannot be switched into an FDL with delay higher than $i$. Otherwise, if there is a departure request while no packet arrives in each of the next $i$ time slots, the packet with the $i$-th highest priority cannot leave the system in time.

 One basic idea to satisfy the above condition is that, set the delays of FDLs as $1,2,4,\ldots,2^{\ell-2},2^{\ell-1},2^{\ell-2},\ldots,4,2,1$, and use a self-routing policy as follows: let packet with tag belonging to $\Psi_j$ enter FDL $j$, where for $j=1,2,\ldots,\ell$,
\begin{equation*}
  \Psi_j=\langle2^{j-1},2^j-1\rangle,
\end{equation*}
and for $j=\ell+1,\ell+2,\ldots,2\ell-1$,
\begin{equation*}
  \Psi_j=\langle3\times 2^{\ell-1}-2^{2\ell-j},3\times 2^{\ell-1}-2^{2\ell-j-1}-1\rangle.
\end{equation*}
The third column of Table~\ref{tab:parameters} gives the values of $\Psi_j$ for $\ell=5$.
(Here the delay sequence and the tag set sequence exhibit a symmetric structure which is employed for the priority loss property.) This setting is ``ideal" in the sense that the switching system can buffer up to $O(2^{\ell})$ packets. However, this setting fails to be a priority queue. The underlying issue is collision, i.e., there will be multiple packets with tags belonging to a same $\Psi_j$  that enter a same FDL at the same time according to the routing policy.

As multiplexers have multiple inputs providing the possibility to solve the collision issue, we are motivated to replace each FDL with a multiplexer with buffer equal to the delay of the FDL. However, this cannot solve the collision issue completely since the number of packets entering a multiplexer can be larger than the number of inputs of the multiplexer (which should be a limited number for construction efficiency). Besides, we need to get rid of buffer overflow at each multiplexer.

To solve the collision issue fundamentally, our key idea is to use multiple multiplexers with smaller buffers as a group to replace each FDL instead of using a single multiplexer. In this way, we can guarantee that the packets entering a group of multiplexers can only come from certain groups of multiplexers except for the arrival link, which have a limited number. So by using multiplexers with a proper number of inputs, the collisions can be avoided. Also, we can establish an upper bound on the number of packets that need to be buffered at some group of multiplexers, and then choose a proper number of multiplexers in a group such that the total buffer size exceeds the upper bound. Thanks to the property that the buffer of a multiplexer is always used in a consecutive manner as mentioned in Sec.~\ref{sec:preliminaries}, buffer overflow can thus never happen at each multiplexer as long as the buffers of the multiplexers in a same group are equally used (differing by at most one packet).

\subsection{Description of the Construction}

Now we formally introduce our construction of optical priority queue.
\subsubsection{Structure}
Let $\ell$ be a positive integer. In our construction, an optical priority queue, as illustrated in Fig.~\ref{fig:detailedconstruction}, consists of a $(24\ell-10)\times (24\ell-10)$ crossbar switch and $2\ell-1$ groups of multiplexers.  For each $j=1,2,\ldots,2\ell-1$, the $j$-th group of multiplexers consists of three parallel 4-to-1 multiplexers with buffer $B_j$, where
\begin{equation*}
  B_j=\begin{cases}
    1 & j=1\\
    2^{j-2} & j=2,3,\ldots,\ell\\
    2^{2\ell-j-2} & j=\ell+1,\ell+2,\ldots,2\ell-2\\
    1 & j=2\ell-1.
  \end{cases}
\end{equation*}
So each group of multiplexers has 12 input links in total. For $i=0,1,2$, we label the 4 input links in the $i$-th multiplexer as $i$-th, $(i+3)$-th, $(i+6)$-th and $(i+9)$-th input links of the group of multiplexers. See Fig.~\ref{fig:groupofmux} for an illustration. The reason for using three multiplexers each with four inputs in a group will be clear after our analysis (c.f. Remark~\ref{remark:why4} and Lemma~\ref{lem:multiplexersize}).

Recall the definition of $\Psi_j$ given in Sec.~\ref{sec:motivation}. We have
\begin{equation*}
  |\Psi_j|=
  \begin{cases}
    B_j & j=1 \text{ or }j=2\ell-1\\
    2B_j & j=2,3,\ldots,2\ell-2.
  \end{cases}
\end{equation*}
Let
\begin{equation*}
  B^*\triangleq 3\times 2^{\ell-1}-2=U(\Psi_{2\ell-1}).
\end{equation*}
Table~\ref{tab:parameters} gives an example on these parameters where $\ell=5$.

\begin{figure}[!tb]
    \centering
\begin{tikzpicture}[>=stealth]
    	\tikzset{myptr/.style={decoration={markings,mark=at position 1 with
   			{\arrow[scale=1,>=stealth]{>}}},postaction={decorate}}}

        \draw (3.5,3.5) rectangle (7.5,7.5);
        \node[scale=2] at (5.5,5.5){Switch};

        \draw[rounded corners] (3.8,2) rectangle (7.2,2.7);
        \node[scale=1.2] at (5.5,2.35){Controller};

        \draw (5.5,3.5)--(6.1,3.1)--(5.75,3.1)--(5.75,2.7);
        \draw (5.5,3.5)--(4.9,3.1)--(5.25,3.1)--(5.25,2.7);

        \draw[myptr](1.7,4.1)--node [pos=0.5,below,sloped]{arrival}(3.5,4.1);

        \draw[myptr](7.5,4.1)--(9.3,4.1);
        \node at(8.5,4.3){departure};

        \draw[myptr](7.5,3.7)--(9.3,3.7);
        \node [below] at(8.5,3.7){loss};

        \draw (4.2,7.7) rectangle (6.8,8.5);
        \node[scale=0.9,align=center] at (5.5,8.1){$1^{st}$ group of\\
        	4-to-1 multiplexers};
        \foreach \x in {1,4}
        {
        	\draw (7.5,7.6-0.2*\x)--(7.6+0.1*\x,7.6-0.2*\x)--(7.6+0.1*\x,7.6+0.2*\x);
        	\draw[myptr] (7.6+0.1*\x,7.6+0.2*\x)--(6.8,7.6+0.2*\x);
        }
    	\foreach \x in {1,2,3}
    	{
    		\draw (4.2,7.7+0.2*\x)--(3.3-0.1*\x,7.7+0.2*\x)--(3.3-0.1*\x,7.5-0.2*\x);
    		\draw[myptr] (3.3-0.1*\x,7.5-0.2*\x)--(3.5,7.5-0.2*\x);
    	}

    \draw (4.2,8.7) rectangle (6.8,9.5);
    \node[scale=0.9,align=center] at (5.5,9.1){$2^{nd}$ group of\\
    	4-to-1 multiplexers};
    \foreach \x in {1,4}
    {
    	\draw (7.5,6.8-0.2*\x)--(8.1+0.1*\x,6.8-0.2*\x)--(8.1+0.1*\x,8.6+0.2*\x);
    	\draw[myptr] (8.1+0.1*\x,8.6+0.2*\x)--(6.8,8.6+0.2*\x);
    }
	\foreach \x in {1,2,3}
	{
		\draw (4.2,8.7+0.2*\x)--(2.9-0.1*\x,8.7+0.2*\x)--(2.9-0.1*\x,6.8-0.2*\x);
		\draw[myptr] (2.9-0.1*\x,6.8-0.2*\x)--(3.5,6.8-0.2*\x);
	}

    \draw [fill] (3.3,5.8) circle [radius=0.025];
    \draw [fill] (3.3,5.6) circle [radius=0.025];
    \draw [fill] (3.3,5.4) circle [radius=0.025];

    \draw [fill] (7.7,5.8) circle [radius=0.025];
    \draw [fill] (7.7,5.6) circle [radius=0.025];
    \draw [fill] (7.7,5.4) circle [radius=0.025];

    \draw [fill] (5.5,9.7) circle [radius=0.025];
    \draw [fill] (5.5,9.9) circle [radius=0.025];
    \draw [fill] (5.5,10.1) circle [radius=0.025];

    \draw [fill] (8.8,7.6) circle [radius=0.025];
    \draw [fill] (9,7.6) circle [radius=0.025];
    \draw [fill] (9.2,7.6) circle [radius=0.025];

    \draw [fill] (2.3,7.6) circle [radius=0.025];
    \draw [fill] (2.1,7.6) circle [radius=0.025];
    \draw [fill] (1.9,7.6) circle [radius=0.025];

    \foreach \x in {1,2,3}{
    	\draw [fill] (7.7,6.1+0.1*\x) circle [radius=0.02];
    	\draw [fill] (7.7,6.9+0.1*\x) circle [radius=0.02];
    	\draw [fill] (7.7,7.9+0.1*\x) circle [radius=0.02];
    	\draw [fill] (7.7,8.9+0.1*\x) circle [radius=0.02];
    	\draw [fill] (7.7,10.5+0.1*\x) circle [radius=0.02];
    }

    \draw (4.2,10.3) rectangle (6.8,11.1);
    \node[scale=0.9,align=center] at (5.5,10.7){($2\ell$-1)-th group of\\
    	4-to-1 multiplexers};
    \foreach \x in {1,4}
    {
    	\draw (7.5,5.4-0.2*\x)--(9.4+0.1*\x,5.4-0.2*\x)--(9.4+0.1*\x,10.2+0.2*\x);
    	\draw[myptr] (9.4+0.1*\x,10.2+0.2*\x)--(6.8,10.2+0.2*\x);
    }
	\foreach \x in {1,2,3}
	{
		\draw (4.2,10.3+0.2*\x)--(1.7-0.1*\x,10.3+0.2*\x)--(1.7-0.1*\x,5.2-0.2*\x);
		\draw[myptr] (1.7-0.1*\x,5.2-0.2*\x)--(3.5,5.2-0.2*\x);
	}

    \node[scale=0.8,right] at (7.8,10.65){12 input links};
    \draw[->](8.5,11.3)--(8.5,11);
    \draw[->](8.5,10.1)--(8.5,10.4);

    \end{tikzpicture}
        \caption{The construction of an optical priority queue based on 4-to-1 multiplexers. Each group of 4-to-1 multiplexers consists of three 4-to-1 multiplexers with same buffer size. Here the loss links of 4-to-1 multiplexers are omitted.}
    \label{fig:detailedconstruction}
\end{figure}
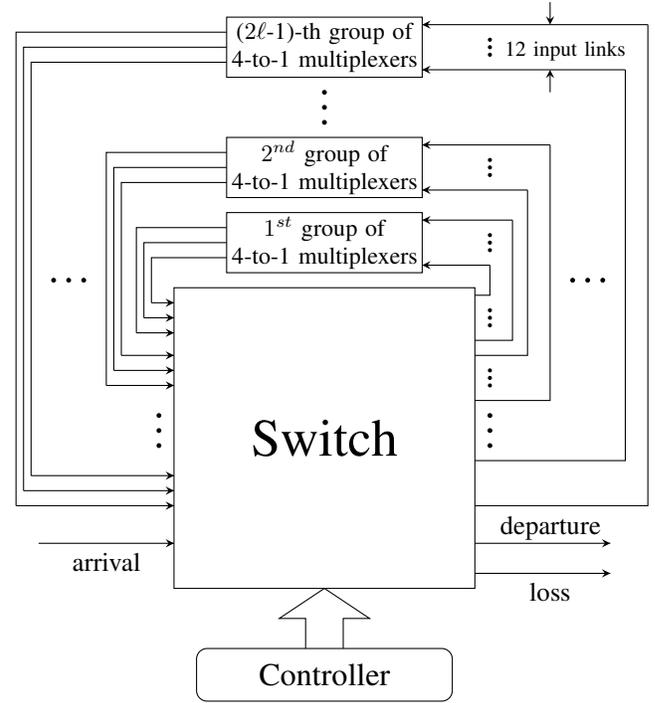

\begin{figure}[!tb]
    \centering
\begin{tikzpicture}[>=stealth]

    \draw (2,2) rectangle (5,3.6);
    \node[scale=1.2] at(3.5,3.1){4-to-1};
    \node[scale=1.2] at(3.5,2.5){multiplexer};
    \draw[->](5,2.8)--(6.5,2.8);
    \foreach \x in {1,2,3,4}
     \draw[->](0.5,1.8+0.4*\x)--(2,1.8+0.4*\x);

    \draw (2,4) rectangle (5,5.6);
    \node[scale=1.2] at(3.5,5.1){4-to-1};
    \node[scale=1.2] at(3.5,4.5){multiplexer};
    \draw[->](5,4.8)--(6.5,4.8);
    \foreach \x in {1,2,3,4}
     \draw[->](0.5,3.8+0.4*\x)--(2,3.8+0.4*\x);

    \draw (2,6) rectangle (5,7.6);
    \node[scale=1.2] at(3.5,7.1){4-to-1};
    \node[scale=1.2] at(3.5,6.5){multiplexer};
    \draw[->](5,6.8)--(6.5,6.8);
    \foreach \x in {1,2,3,4}
     \draw[->](0.5,5.8+0.4*\x)--(2,5.8+0.4*\x);

    \node[left,scale=0.9] at(0.5,7.4) {0};
    \node[left,scale=0.9] at(0.5,7.0) {3};
    \node[left,scale=0.9] at(0.5,6.6) {6};
    \node[left,scale=0.9] at(0.5,6.2) {9};

    \node[left,scale=0.9] at(0.5,5.4) {1};
    \node[left,scale=0.9] at(0.5,5.0) {4};
    \node[left,scale=0.9] at(0.5,4.6) {7};
    \node[left,scale=0.9] at(0.5,4.2) {10};

    \node[left,scale=0.9] at(0.5,3.4) {2};
    \node[left,scale=0.9] at(0.5,3.0) {5};
    \node[left,scale=0.9] at(0.5,2.6) {8};
    \node[left,scale=0.9] at(0.5,2.2) {11};

    \draw[dashed] (1.2,1.8) rectangle (5.8,7.8);
    \end{tikzpicture}
     \caption{A group of 4-to-1 multiplexers consists of three 4-to-1 multiplexers. The indices of 12 input links are labelled. Here the loss links are omitted.}
    \label{fig:groupofmux}
\end{figure}
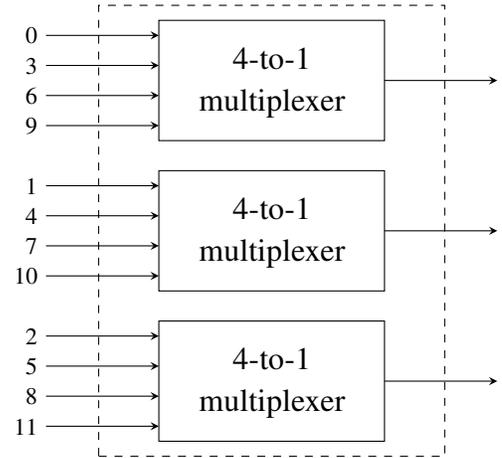

\subsubsection{Routing Policy}
The routing policy performed by the switch at the beginning of time $t$, $t=1,2,\ldots$, is as follows.
\begin{itemize}
  \item If there is a departure request, i.e., $c(t)=1$, then compare the packets out from the first and the second groups of multiplexers as well as the arriving packet, if any, and let the one having the highest priority depart this switching system from the departure link.

  \item If there is no departure request and there is an arriving packet, but the buffer is full, i.e., $c(t)=0$, $a(t)=1$ and $q(t-1)=B^*$, then compare the arriving packet and the packets out from the last group of multiplexers, if any, and let the packet having the lowest priority leave the switching system from the loss link.

  \item Every other packet $i$ entering the switch will be pushed into the $j$-th group of multiplexer such that
  \begin{equation}
  \label{eq:routing}
    \tau_i(t)\in \Psi_j.
  \end{equation}
  Specifically, suppose that there are $k$ packets entering the $j$-th group of multiplexer according to \eqref{eq:routing}, and let $u_j(t)$ be the index of the input link of the group that is lastly used before $t$, where $u_j(1)=0$. Then these $k$ packets will enter the $j$-th group via $((u_j(t)+1)\% 12)$-th, $((u_j(t)+2)\% 12)$-th, \ldots, $((u_j(t)+k)\% 12)$-th input links, respectively. Also, set $u_j(t+1)=(u_j(t)+k)\% 12$. Here $x\% y$ denotes the reminder of $x$ when divided by $y$.
%
\end{itemize}

It is straightforward to see that, the computation cost of the above routing policy at each time is linear with the number of packets entering the switch at that time, or equivalently, $O(\ell)$.  It is remarkable that this routing policy is feasible only if there is no packet collision at each group of multiplexer, i.e., the number of packets entering a group of multiplexers at each time is at most 12, the total number of input links of the group of multiplexers. As we will show in Lemma~\ref{lem:nocollision}, this requirement can always be satisfied.

We have the following main result.
\begin{theorem}
    \label{thm:emulation}
  The proposed switching system is an optical priority queue with buffer $B^*$.
\end{theorem}

The proof of Theorem~\ref{thm:emulation} is deferred to Sec.~\ref{sec:proof}.

\begin{table}[!tb]
 \center
\caption{Some parameters of an optical construction of priority queues where $\ell=5$ and $B^*=46$.. The column of tag range is due to Lemma~\ref{lem:range}, and the last column is due to Lemma~\ref{lem:multiplexersize}.}
\label{tab:parameters}
\begin{tabular}{ccccc}
\toprule
$j$ & $B_j$ & $\Psi_j$ & tag range & num. of buffered pkt\\

\midrule
\rowcolor[gray]{0.9} 1 & 1 & $\{1\}$ & $\{1\}$ & $\leq 1$ \\

2 & 1 & $\langle 2,3\rangle$ & $\langle 2,3\rangle$ & $\leq 2$\\

\rowcolor[gray]{0.9} 3 & 2 & $\langle 4,7\rangle$ & $\langle 3,8\rangle$ & $\leq 5$ \\

4 & 4 & $\langle 8,15\rangle$ & $\langle 5,18\rangle$ & $\leq 11$\\

\rowcolor[gray]{0.9} 5 & 8 & $\langle 16,31\rangle$ & $\langle 9,38\rangle$ & $\leq 23$  \\

6 & 4 & $\langle 32,39\rangle$ & $\langle 29,42\rangle$ & $\leq 11$\\

\rowcolor[gray]{0.9} 7 & 2 & $\langle40,43\rangle$ & $\langle 39,44\rangle$ & $\leq 5$\\

8 & 1 & $\langle 44,45\rangle$ & $\langle 44,45\rangle$ & $\leq 2$\\

\rowcolor[gray]{0.9} 9 & 1 & $\{46\}$ &$\{46\}$ & $\leq 1$ \\
\bottomrule
\end{tabular}
\end{table}

\subsection{Construction Cost}

The cost of our construction of optical priority queues depends on how to construct 4-to-1 multiplexers with SDLs. As will be demonstrated in Lemma~\ref{lem:nooverflow} in Sec.~\ref{sec:proof}, there would never be any buffer overflow at each multiplexer. Based on this fact, some requirements on the used 4-to-1 multiplexers could be relaxed.
\begin{itemize}
  \item First, any 4-to-1 multiplexer with $B_j$ buffer could be replaced by a 4-to-1 multiplexer with buffer larger than or equal to $B_j$. This is because, when no buffer overflow happens at either of them, they have identical departure processes if both are started from empty systems and subject to identical arrival processes.

 \item Second, a 4-to-1 multiplexer could be replaced by a 4-to-1 \emph{delayed-loss} multiplexer with the same buffer size. A 4-to-1 delayed-loss multiplexer and a 4-to-1 multiplexer with the same buffer size have identical departure processes if they are started from empty systems and subject to identical arrival processes. The only difference is that the loss processes of the two systems do not match exactly. See~\cite{chang2004recursive} for a formal definition of delayed-loss multiplexers.
\end{itemize}

Regarding the construction of delayed-loss multiplexers, we have the following result based on the construction proposed by
Chang \emph{et al.}~\cite{chang2004recursive}.
\begin{lemma}
\label{lem:4to1}
  For any positive integer $k$, an $n$-to-1 delayed-loss multiplexer with buffer $n^k-1$ can be constructed with a $((n-1)k+n)\times ((n-1)k+n)$ crossbar switch and $(n-1)k$ FDLs.
\end{lemma}
\begin{proof}
  Chang \emph{et al.}~\cite{chang2004recursive} gave a construction of an $n$-to-1 delayed-loss multiplexer with buffer $n^k-1$, which consists of $k+1$ $n\times n$ crossbar switches, indexed by $0,1,\ldots,k$, in tandem. More specifically, the $n$ input links of the $0$-th crossbar switch act as the $n$ input links of the $n$-to-1 delayed-loss multiplexer, and the $n$ output links of the $k$-th crossbar switch act as the output link and $n-1$ loss links of the multiplexer. For $i=0,1,\ldots,k-1$, the $n$ output links of the $i$-th crossbar switch connect to $n$ input links of the $i+1$-th crossbar switch, each via an FDL with some specific delay except one via a direct link. See \cite[Fig. 17]{chang2004recursive} for an illustration. This construction uses $(n-1)k$ FDLs in total.

 Now consider the integration of all the switches in the construction into one. A straightforward integration will consist of an $(nk+n)\times (nk+n)$ crossbar switch, $(n-1)k$ FDLs which connect $(n-1)k$ outputs and $(n-1)k$ inputs of the switch, and $k$ direct links which connect $k$ outputs and $k$ inputs of the switch. Note that the $k$ direct links become useless in this integration. So the links together with the corresponding inputs and outputs can be removed from this integration, which leads to a construction of $((n-1)k+n)\times ((n-1)k+n)$ crossbar switch and $(n-1)k$ FDLs.
 \end{proof}

Define
\begin{equation*}
  B_j'=\begin{cases}
    3 & j=1\\
    4^{\lceil\frac{j-1}{2}\rceil}-1 & j=2,3,\ldots,\ell\\
    4^{\lceil \frac{2\ell-j-1}{2}\rceil}-1 & j=\ell+1,\ell+2,\ldots,2\ell-2\\
    3 & j=2\ell-1.
  \end{cases}
\end{equation*}
It is straightforward to check that $B_j'\geq B_j$ for all $j=1,2,\ldots,2\ell-1$.

In order to take advantage of Lemma~\ref{lem:4to1}, we replace each 4-to-1 multiplexer with buffer $B_j$, $j=1,2,\ldots,2\ell-1$ in our construction with a 4-to-1 delayed-loss multiplexer with buffer $B_j'$. We refer to this construction as \emph{specialized construction}. According to our analysis, the specialized construction is also an optical priority queue with buffer $B^*$. By further integrating all the switches used into one, we have the following result (the result for the trivial case that $\ell=1$ is omitted).

\begin{theorem}
\label{thm:constructioncost}
For any positive integer $\ell\geq 2$, an optical priority queue with $3\times 2^{\ell-1}-2$ can be constructed with a $(\frac{1}{2}(9\ell^2+39\ell)+8)\times (\frac{1}{2}(9\ell^2+39\ell)+8)$ crossbar switch and $\frac{9}{2}(\ell^2-\ell)+18$ FDLs.
\end{theorem}
\begin{proof}
 We consider the specialized construction where we adopt the method in Lemma~\ref{lem:4to1} to construct the 4-to-1 delayed-loss multiplexers. According to Lemma~\ref{lem:4to1}, the number of FDLs used by this construction is
  \begin{align*}
    &3\left (3+\sum_{j=2}^{\ell}3\left \lceil \frac{j-1}{2}\right  \rceil +\sum_{j=\ell+1}^{2\ell-2}3\left \lceil \frac{2\ell-j-1}{2}\right \rceil+3\right) \\
    =& \frac{9}{2}(\ell^2-\ell)+18,
  \end{align*}
which holds for both $\ell$ is even and $\ell$ is odd.

Note that when integrating all the switches into one, the size of the integrated switch is the sum of the sizes of all the switches used for constructing 4-to-1 delayed-loss multiplexers plus 2 other than plus $24\ell-10$. According to Lemma~\ref{lem:4to1}, this is equal to
\begin{align*}
  &3\Bigg(7+\sum_{j=2}^{\ell}\left(3\left \lceil \frac{j-1}{2}\right  \rceil+4\right) \\
  &\qquad  +\sum_{j=\ell+1}^{2\ell-2}\left(3\left \lceil \frac{2\ell-j-1}{2}\right \rceil+4\right)+7\Bigg)+2\\
    =& \frac{1}{2}(9\ell^2+39\ell)+8.
\end{align*}
The proof is accomplished.
\end{proof}

Considering the construction framework depicted in Fig.~\ref{fig:general}, we have the following result.
\begin{theorem}
  There exists a construction of optical priority queue with buffer $2^{\Theta(\sqrt{M})}$ using a single $(M+2)\times (M+2)$ crossbar switch and $M$ FDLs.
\end{theorem}
\begin{proof}
  Let $M=\frac{1}{2}(9\ell^2+39\ell)+6$. Then $\ell=\Theta (\sqrt{M})$ and $B^*=3\times 2^{\ell-1}-2=2^{\Theta(\sqrt{M})}$. According to Theorem~\ref{thm:constructioncost}, this result holds directly.
\end{proof}

\begin{remark}
  It remains open that how to construct a $4$-to-1 multiplexer with an arbitrary size optimally in the sense of the number of FDLs used. Any more efficient method for constructing $4$-to-1 multiplexers with buffer size $B_j$ than the one for constructing a 4-to-1 delayed-loss multiplexers with buffer size $B_j'$ would lead to a better result than Theorem~\ref{thm:constructioncost}.
\end{remark}

\begin{remark}
  The construction cost given in Theorem~\ref{thm:constructioncost} can be reduced by, e.g., replacing the first/last group of multiplexers by a single FDL, replacing each multiplexer in the second/last second group by a single FDL, etc. But these changes can only reduce the construction cost by a small fixed number, which does not change the result in the order sense.
\end{remark}

%
%
%
%





\section{Proof of Theorem~\ref{thm:emulation}}
\label{sec:proof}

According to Definition~\ref{def:pq}, we need to show that the proposed switching system satisfies all the properties (P1)-(P5) for any time $t>0$. We will prove this by induction on time $t$. It is straightforward to check that these properties hold in the base case $t=1$. For induction, we assume that the proposed switching system satisfies all the properties (P1)-(P5) for every time $t<T$. We will show that it also satisfies these properties for $t=T$.
%
%
%
%

%
%

The proof proceeds as follows. First, we will give some basic results about the changing of the tag of a packet. Then, we will prove that the proposed routing policy is collision-free, which guarantees the feasibility of the routing policy. Later, we will show that there is no buffer overflow at each multiplexer. Finally, we will prove Theorem~\ref{thm:emulation} based on the preparations.

\subsection{Tag Changing}

We first show that the tag of any packet in the switching system can change by at most one in a time slot.
\begin{lemma}
\label{lem:tagchanging}
  For any packet $i$ in this switching system at both time $t$ and time $t-1$ where $t\leq T$,
  \begin{equation}
  \label{eq:tagchanging}
    |\tau_{i}(t)-\tau_{i}(t-1)|\leq 1,
  \end{equation}
\end{lemma}

\begin{proof}
Since the switching system emulates a priority queue up to time $T-1$, we can show the result based on the properties of a priority queue. We consider two cases:

Case 1: there is no arriving packet at time $t$. As (P1)-(P5) hold for time $t-1$, it is straightforward to see that $\tau_i(t)=\tau_i(t-1)-1$ if there exists a departure packet at time $t-1$, or $\tau_i(t)=\tau_i(t-1)$ if otherwise. Hence, \eqref{eq:tagchanging} holds.

Case 2: there is an arriving packet at time $t$. If the packet has a lower priority than $i$, then the argument for case 1 also holds. If the packet has a higher priority than $i$, then $\tau_i(t)=\tau_i(t-1)$ if there exists a departure packet at time $t-1$, or $\tau_i(t)=\tau_i(t-1)+1$ if otherwise. For all these subcases, \eqref{eq:tagchanging} holds.
\end{proof}

Lemma~\ref{lem:tagchanging} directly implies the following result, which is a generalization of Lemma~\ref{lem:tagchanging}.
\begin{corollary}
\label{cor:tagchanging}
  For any packet $i$ in the switching system at both time $t$ and time $t'$ where $t'<t\leq T$,
  \begin{equation*}
    |\tau_i(t)-\tau_i(t')|\leq t-t'.
  \end{equation*}
\end{corollary}

\subsection{Collision-free}

We first show the range of the tag of a packet buffered at some group of multiplexers.
\begin{lemma}
  \label{lem:range}
  For any packet $i$ buffered at the $j$-th group of multiplexers at time $t< T$,
  \begin{equation*}
    L(\Psi_{j})-B_j+1\leq \tau_i(t)\leq U(\Psi_{j})+B_j-1.
  \end{equation*}
\end{lemma}
\begin{proof}
  Consider a packet $i$ buffered at some multiplexer in the $j$-th group of multiplexers at time $t<T$. Let $t'\leq t$ be time that $i$ entered the multiplexer for the last time. According to properties (M2) and (M4) of multiplexers, $i$ would depart from the multiplexer in at most $B_j$ time steps since $t'$. Hence, $t-t'\leq B_j-1$. By Corollary~\ref{cor:tagchanging}, we have
  \begin{equation*}
    |\tau_i(t)-\tau_i(t')|\leq t-t'\leq B_j-1.
  \end{equation*}
  This completes the proof as $\tau_i(t')\in \Psi_j$ according to the routing policy.
\end{proof}

The following result shows that the proposed routing policy is collision-free, which guarantees the feasibility of the proposed routing policy.
\begin{lemma}
\label{lem:nocollision}
For any $j$, the number of packets entering the $j$-th group of multiplexers at time $T$ under the routing policy is at most 10.
\end{lemma}
\begin{proof}
Consider an arbitrary packet $i$ that is buffered at some multiplexer in the $j$-th group at time $T-1$, but leaves the multiplexer and enters the switch at $T$.

We have the following claim: for any $2\leq j\leq 2\ell-1$,
\begin{equation*}
  \tau_i(T)\geq L(\Psi_{j-1}),
\end{equation*}
and for any $1\leq j\leq 2\ell-2$,
\begin{equation*}
  \tau_i(T)\leq U(\Psi_{j+1}).
\end{equation*}
According to the routing policy, this implies that, $i$ can only enter the $(j-1)$-th group of multiplexers, $j$-th group of multiplexers, or $(j+1)$-th group of multiplexers at $T$, if exists. In other words, the packets entering the $j$-th group of multiplexers at time $T$ can only come from the packets leaving the $(j-1)$-th group of multiplexers, the $j$-th group of multiplexers, the $(j+1)$-th group of multiplexers at $T-1$, or the arrival link. Since only one packet can depart from a multiplexer at a time, the number of packets entering the $j$-th group of multiplexers at time $T$ under the routing policy is at most 10.

In the following, we will prove the claim. By Lemma~\ref{lem:range}, we have
\begin{equation*}
  L(\Psi_j)-B_j+1\leq \tau_i(T-1)\leq U(\Psi_j)+B_j-1.
\end{equation*}
Hence, for $2\leq j\leq 2\ell-1$,
\begin{IEEEeqnarray*}{rCl}
  \tau_i(T)&\geq & L(\Psi_j)-B_j\\
  &\geq & L(\Psi_j)-|\Psi_{j-1}|\\
  &=& L(\Psi_{j-1}),
\end{IEEEeqnarray*}
where the first inequality holds according to Lemma~\ref{lem:tagchanging}, and the second inequality holds since $B_j=|\Psi_{j-1}|$ if $2\leq j\leq \ell$, $B_j=|\Psi_{j-1}|/4$ if $\ell+1\leq j\leq 2\ell-2$, and $B_j=|\Psi_{j-1}|/2$ if $j=2\ell-1$. Similarly, for $1\leq j\leq 2\ell-2$,
\begin{IEEEeqnarray*}{rCl}
  \tau_i(T)&\leq & U(\Psi_j)+B_j\\
  &\leq & U(\Psi_j)+|\Psi_{j+1}|\\
  &=& U(\Psi_{j+1}),
\end{IEEEeqnarray*}
where the first inequality holds according to Lemma~\ref{lem:tagchanging}, and the second inequality holds since $B_j=|\Psi_{j+1}|/2$ if $j=1$, $B_j=|\Psi_{j+1}|/4$ if $2\leq j\leq \ell-1$ and $B_j=|\Psi_{j+1}|$ if $j\geq \ell$.
The proof is accomplished.
\end{proof}

\begin{remark}
\label{remark:why4}
  It is worth mentioning that, the proof of Lemma~\ref{lem:nocollision} is independent with the number of inputs of each multiplexer. In order to be collision-free, the total number of inputs of multiplexers is required to be larger than or equal to 10. Meanwhile, in order to guarantee the buffers of multiplexers in a group are equally used~(c.f. Lemma~\ref{lem:balance}), we should let these multiplexers have the same number of inputs. Hence, the number of inputs of multiplexers should be at least 4. In general, the construction cost of an $n$-to-1 multiplexer with a fixed buffer size grows larger with $n$~(c.f. Lemma~\ref{lem:4to1}).  So our design uses 4-to-1 multiplexers for construction efficiency.
\end{remark}

\subsection{No Buffer Overflow}

In the following, we will show that there would not be any buffer overflow at each multiplexer. We start by showing that the difference between the tags of any pair of packets in the switching system can change by at most 1 in a time slot.
\begin{lemma}
\label{lem:diff}
  For any $t\leq T$ and any packets $i_1$ and $i_2$ in this switching system at both time $t$ and time $t-1$,
  \begin{equation}
  \label{eq:diff}
   |(\tau_{i_1}(t)-\tau_{i_2}(t))-(\tau_{i_1}(t-1)-\tau_{i_2}(t-1))|\leq 1.
  \end{equation}
\end{lemma}
\begin{proof}
 We consider all the four possible cases:

  Case 1: there is no packet arriving at the switching system at $t$. If there exists a departure packet at time $t-1$, then $\tau_{i_1}(t)=\tau_{i_1}(t-1)-1$ and $\tau_{i_2}(t)=\tau_{i_2}(t-1)-1$. Otherwise, $\tau_{i_1}(t)=\tau_{i_1}(t-1)$ and $\tau_{i_2}(t)=\tau_{i_2}(t-1)$. For both of the subcases, \eqref{eq:diff} holds.

  Case 2: there is a packet arriving at the switching system at $t$, which has a lower priority than both $i_1$ and $i_2$. Clearly, the argument for case 1 also holds.

  Case 3: there is a packet arriving at the switching system at $t$, which has a higher priority than both $i_1$ and $i_2$. If there exists a departure packet at time $t-1$, then $\tau_{i_1}(t)=\tau_{i_1}(t-1)$ and $\tau_{i_2}(t)=\tau_{i_2}(t-1)$. Otherwise, $\tau_{i_1}(t)=\tau_{i_1}(t-1)+1$ and $\tau_{i_2}(t)=\tau_{i_2}(t-1)+1$. Hence, \eqref{eq:diff} holds.

  Case 4: there is a packet arriving at the switching system at $t$, which has a higher priority than $i_1$ but lower than $i_2$ (without loss of generality, we here assume $i_1$ has a lower priority than $i_2$). If there exists a departure packet at time $t-1$, then $\tau_{i_1}(t)=\tau_{i_1}(t-1)$ and $\tau_{i_2}(t)=\tau_{i_2}(t-1)-1$. Otherwise, $\tau_{i_1}(t)=\tau_{i_1}(t-1)+1$ and  $\tau_{i_2}(t)=\tau_{i_2}(t-1)$. Hence, \eqref{eq:diff} also holds in this case.
\end{proof}

From Lemma~\ref{lem:range}, we can see that the number of packets buffered at the $j$-th group of multiplexer is at most $|\Psi_j|+2B_j-2$, which is equal to $3B_j-2$ if $j=1$ or $2\ell-1$, or $4B_j-2$ if $2\leq j\leq 2\ell-2$. This bound can be improved to $3B_j-2$ for any $j$ by the following result.

\begin{lemma}
\label{lem:multiplexersize}
  For any two packets $i_1$ and $i_2$ that are buffered at, or entering the $j$-th group of multiplexers at time $T$,
    \begin{equation*}
    |\tau_{i_1}(T)-\tau_{i_2}(T)|\leq 3B_j-2.
  \end{equation*}
\end{lemma}

\begin{proof}
  Suppose that the time that packets $i_1$ and $i_2$ entered the $j$-th group of multiplexers for the last time before $T$ or at $T$ is $t_1$ and $t_2$. Without loss of generality, we assume that $t_2\leq t_1$. By Lemma~\ref{lem:diff}, we have
  \begin{equation}
  \label{eq:a}
    |\tau_{i_1}(T)-\tau_{i_2}(T)|\leq |\tau_{i_1}(t_1)-\tau_{i_2}(t_1)|+(T-t_1).
  \end{equation}
  By Corollary~\ref{cor:tagchanging}, we also have
  \begin{equation}
  \label{eq:b}
    |\tau_{i_2}(t_1)-\tau_{i_2}(t_2)|\leq t_1-t_2.
  \end{equation}
According to the routing policy, $\tau_{i_1}(t_1)\in \Psi_j$ and $\tau_{i_2}(t_2)\in \Psi_j$.  Since $\Psi_j$ consists of consecutive integers,
\begin{equation}
\label{eq:c}
  |\tau_{i_1}(t_1)-\tau_{i_2}(t_2)|\leq |\Psi_j|-1\leq 2B_j-1.
\end{equation}
Combining \eqref{eq:a}, \eqref{eq:b} and \eqref{eq:c}, we have
\begin{equation*}
 |\tau_{i_1}(T)-\tau_{i_2}(T)|\leq T-t_2+2B_j-1.
  \end{equation*}
  Note that $T-t_2\leq B_j-1$ since otherwise $i_2$ would leave the $j$-th group of multiplexers before $T$. This completes the proof.
\end{proof}

For $t<T$, let $q_j(i,t), i=0,1,2$ be the number of packets buffered at the $i$-th multiplexer in the $j$-th group of multiplexers at time $t$, and let $q_j(i,T), i=0,1,2$ be the number of packets buffered at or entering the $i$-th multiplexer in the $j$-th group of multiplexers at time $T$. Recall that, according to the routing policy, the input links of a group of multiplexers are used in a round-robin manner. Based on this scheme and together with the non-idling property (M2) of multiplexers, we can show that the buffers of the multiplexers in a same group are always almost equally used, i.e., the number of packets buffering in the multiplexers differs by at most one. Besides, if the input link that is lastly used before $t$ belongs to the $i$-th multiplexer in the group, then at time $t-1$, the number of packets buffered in the $i$-th multiplexer at time $t-1$ is equal to or larger than that in the $((i-1)\%3)$-th multiplexer, which is also equal to or larger than that in the $((i-2)\%3)$-th multiplexer.

\begin{lemma}
\label{lem:balance}
  Consider any $j$-th group of multiplexers, and let $k(t)=u_j(t+1) \% 3$. Then, for $t\leq T$,
  \begin{equation}
  \label{eq:1}
  q_j(k(t),t)\geq q_j((k(t)-1)\% 3,t)\geq q_j((k(t)-2)\% 3,t),
\end{equation}
and
\begin{equation}
\label{eq:2}
  q_j(k(t),t)- q_j((k(t)-2)\% 3,t)\leq 1.
\end{equation}
\end{lemma}
\begin{proof}
  We prove this result by induction on $t$. If $t=0$, then $k(t)=0$, and $q_j(0,0)=q_j(1,0)=q_j(2,0)=0$. So, \eqref{eq:1} and \eqref{eq:2} hold.

  Now suppose \eqref{eq:1} and \eqref{eq:2} hold when $t=t_0-1\leq T-1$. We will show that they also hold for $t=t_0$.
  Without loss of generality, we assume that $k(t_0-1)=0$ (the cases that $k(t_0-1)=1$ and $k(t_0-1)=2$ can be considered in a same way). By induction hypothesis, we have
  \begin{equation}
  \label{eq:3}
  q_j(0,t_0-1)\geq q_j(2, t_0-1)\geq q_j(1,t_0-1)
  \end{equation}
  and
  \begin{equation}
  \label{eq:4}
  q_j(0,t_0-1)-q_j(1,t_0-1)\leq 1.
  \end{equation}
   Let $m(i,t_0)$, $i=0,1,2$, denote the number of packets arriving at the $i$-th multiplexer in the group at time $t_0$.
   Then,
    \begin{align*}
    k(t_0)=& u_j(t_0+1)\% 3\\
    =& \left ( u_j(t_0)+\sum_{i=0}^2 m(i,t_0)\right)\% 3\\
    =&\left(\sum_{i=0}^2 m(i,t_0)\right)\% 3
  \end{align*}
  where the second equality holds according to the routing policy and the last equality holds since $u_j(t_0)\% 3=k(t_0-1)=0$.
Since the routing policy uses the multiplexers in a same group in a round robin manner, we have
   \begin{equation}
   \label{eq:aa}
     m(k(t_0),t_0)\geq m((k(t_0)-1)\% 3, t_0) \geq m((k(t_0)-2)\% 3, t_0),
   \end{equation}
   and
   \begin{equation}
   \label{eq:bb}
     m(k(t_0),t_0)-m((k(t_0)-2)\% 3, t_0)\leq 1.
   \end{equation}
   Since the switching system emulates the priority queue for each $t<T$, there is no packet lost at time $t_0$ if $t_0<T$. Hence, according to (M1) and (M2),
  \begin{equation}
  \label{eq:5}
    q_j(i,t_0)=[q_j(i,t_0-1)+m(i,t_0)-1]^+, i=0,1,2.
  \end{equation}
  Clearly, the above equality also holds for $t_0=T$ due to the definition of $q_j(i,T)$.
  We consider three possible cases:

  Case 1: $\sum_{i=0}^2 m(i,t_0)\% 3=0$. Then $k(t_0)=0$. By \eqref{eq:aa} and \eqref{eq:bb},
   $m(0,t_0)=m(1,t_0)=m(2,t_0)$. From \eqref{eq:3}, \eqref{eq:4} and \eqref{eq:5}, it is straightforward to see that $q_j(0,t_0)\geq q_j(2,t_0)\geq q_j(1,t_0)$ and $q_j(0,t_0)-q_j(1,t_0)\leq 1$.

  Case 2: $\sum_{i=0}^2 m(i,t_0)\% 3=1$. Then $k(t_0)=1$. By \eqref{eq:aa} and \eqref{eq:bb}, $m(1,t_0)=m(0,t_0)+1=m(2,t_0)+1$, and. From \eqref{eq:3}, \eqref{eq:4} and \eqref{eq:5}, we can get $q_j(1,t_0)\geq q_j(0,t_0)\geq q_j(2,t_0)$ and $ q_j(1,t_0)-q_j(2,t_0)\leq 1$.

  Case 3: $\sum_{i=0}^2 m(i,t_0)\% 3=2$. Then $k(t_0)=2$. By \eqref{eq:aa} and \eqref{eq:bb}, $m(1,t_0)=m(2,t_0)=m(0,t_0)+1$, and $k(t_0)=2$. From \eqref{eq:3}, \eqref{eq:4} and \eqref{eq:5}, we have $q_j(2,t_0)\geq q_j(1,t_0)\geq q_j(0,t_0)$ and $ q_j(2,t_0)-q_j(0,t_0)\leq 1$.

Hence, for each case, we have \eqref{eq:1} and \eqref{eq:2} for $t=t_0$. By mathematical induction, \eqref{eq:1} and \eqref{eq:2} hold for $t\leq T$.
\end{proof}

%
%

\begin{lemma}
\label{lem:nooverflow}
Any packet arriving at any multiplexer in the switching system at time $T$ cannot be lost.
\end{lemma}
\begin{proof}
By contradiction, we assume that there exists some packet arriving at $i$-th multiplexer in the $j$-th group lost due to overflow. According to (M3), $q_j(i,T)>B_j$. By Lemma~\ref{lem:balance}, this implies that $q_j(i',T)\geq B_j$ for $i'\neq i, i'\in \{1,2,3\}$. Hence, $q_j(1,T)+q_j(2,T)+q_j(3,T)>3B_j$. On the other hand, Lemma~\ref{lem:multiplexersize} implies $q_j(1,T)+q_j(2,T)+q_j(3,T)\leq 3B_j-1$, which leads to a contradiction. The proof is accomplished.
\end{proof}

%
%

\subsection{Completing the Proof}

Now we complete the proof of Theorem~\ref{thm:emulation}. We will show that all the five properties (P1)-(P5) hold at time $T$. First, according to Lemma~\ref{lem:nocollision} and Lemma~\ref{lem:nooverflow}, (P1) holds directly.

To prove (P2) and (P4), we can assume, without loss of generality, that $c(T)=1$ and $q(T-1)+a(T)>0$. Consider the packet $i$ that $\tau_i(T)=1$. If it is the arriving packet, then according to the routing policy, (P2) and (P4) hold directly. If otherwise, $\tau_i(T-1)=1$ or $\tau_i(T-1)=2$ according to Lemma~\ref{lem:tagchanging}. By Lemma~\ref{lem:range}, we can check that $i$ must be buffered at the first group of multiplexers or at the second group of multiplexers at time $T-1$. Recall that the buffer size of each multiplexer in the first group or in the second group is just one. By property (M2), packet $i$ will leave the corresponding multiplexer and enter the switch at time $T$. Hence, according to the routing policy, (P2) and (P4) hold.

(P3) and (P5) can be proved similar to (P2) and (P4). Suppose that there is no departure request and there is an arriving packet at time $T$, while the number of packets buffered in the switching system at time $T-1$ is $B^*$. Consider the packet $i$ that $\tau_i(T)=B^*+1$. If $i$ is the arriving packet, then according to the routing policy, $i$ will be dropped via the loss link at $T$. Hence, (P3) and (P5) hold. If otherwise, then according to Lemma~\ref{lem:tagchanging}, $\tau_i(T-1)=B^*$. By Lemma~\ref{lem:range}, $i$ was buffered at the last group of multiplexers. Recall that the buffer size of each multiplexer in the last group is just one. By property (M2), packet $i$ will leave the corresponding multiplexer at time $T$. According to the routing policy, $i$ will be dropped via the loss link at $T$. Hence, (P3) and (P5) hold at $T$.

The whole proof is accomplished.

\section{Related Work}
\label{sec:relatedwork}

Many methods have been developed for using the SDL-based constructions to exactly emulate various electronic queue structures. Here we introduce the constructions of some typical SDL-based optical components.
\begin{itemize}
  \item \emph{FIFO multiplexers:} In \cite{cruz1996cod}, a design named COD (Cascaded Optical Delay-lines) was proposed for exactly emulating 2-to-1 FIFO multiplexers by using $2\times 2$ crossbar switches and FDLs. However, the number of switches in COD is linear in the buffer size. An improved design named Logarithm Delay-Line Switched was proposed in \cite{Hunter19972} where the number of $2\times 2$ switches used is only logarithmic in the buffer size. In \cite{chang2004recursive}, a recursive construction of 2-to-1 multiplexer was introduced, which was further extended to constructing $n$-to-1 multiplexers using self-routing. In \cite{chou2006necessary}, it was proposed that an $(M+2)\times (M+2)$ crossbar switch and $M=O(\log B)$ FDLs are sufficient to emulate a 2-to-1 multiplexer with buffer $B$. Based on these works, some other constraints including fault-tolerance~\cite{cheng2007constructions}, variable length burst~\cite{chang2006using}, and limited number of recirculations~\cite{cheng2008constructions} are also taken into account for constructing 2-to-1 multiplexers. Since FIFO multiplexers admit efficient SDL based constructions and have some salient properties that FDLs do not have, they can be exploited in the design of optical priority queues, as firstly demonstrated in this work.

  \item \emph{FIFO and LIFO queues:} In \cite{chang2006constructions}, a recursive construction for FIFO queue was  proposed which uses $2\log_2 B-1$ FDLs, where $B$ is the buffer size. In \cite{cheng2013necessary}, a necessary and sufficient condition was characterized for SDL constructions of FIFO queues. In \cite{small2007modular}, a cascade optical LIFO queue architecture based on multiple building-block modules was developed, but its capacity of each module is highly limited. In \cite{huang2007recursive}, the idea of two-level caching was proposed, based on which recursive constructions of parallel FIFO and LIFO queues are proposed. The result in \cite{huang2007recursive} indicate that a LIFO queue of size $B$ can be constructed using at most $9\log_2 B$ FDLs. An improved design was proposed in \cite{wang2011efficient}, which only uses approximately $3\log_2 B$ FDLs. Although FIFO and LIFO queues can be viewed as special cases of priority queues, existing ideas for constructing FIFO and LIFO queues cannot be easily extended for constructing priority queues.

  \item \emph{Priority queues:} In \cite{sarwate2006exact}, Sarwate and Anatharam firstly considered the SDL-based construction of optical priority queues. They showed the buffer size is upper bounded by $2^M+1$, where $M$ is the number of FDLs, and gave a construction of an optical priority queue with $\Theta(M^2)$ buffer. A more general construction framework based on the notion of complementary priority queue was proposed in \cite{chiu2007simple}. Using this framework, an improved design of optical priority queue with $\Theta(M^3)$ buffer was proposed in \cite{chiu2007using}. These results were extended to the construction of optical priority queues with multiple inputs and multiple outputs in \cite{cheng2011constructions}. Very recently, a recursive construction of optical priority queue was proposed which can achieve a buffer size of $\Theta(M^c)$ for any positive integer $c$. All these constructions considered the exact emulation of optical priority queues. In contrast,  ``strong" emulation of optical priority queue was considered in \cite{kogan2007optimal} where each packet departs from the construction with bounded delay.
\end{itemize}

\section{Conclusion}
\label{sec:conclusion}

We have proposed a novel construction of an optical priority queue with buffer $2^{\Theta (\sqrt{M})}$ using a single optical crossbar switch and $M$ FDLs, which leverages 4-to-1 multiplexers for feeding back packets to the switch, and adopts a routing policy that is similar to self-routing. This is a substantial improvement over all previous constructions of optical priority queues which only have polynomial-size buffers. In the future, we would make further efforts towards closing the remaining gap between the exponential upper bound in \cite{sarwate2006exact} and the established sub-exponential lower bound for the SDL design of priority queues. We would also like to see whether our method can be extended to achieve better designs of other network elements (e.g., optical priority queues with multiple inputs and multiple outputs~\cite{cheng2011constructions}).

\ifCLASSOPTIONcaptionsoff
  \newpage
\fi



%
\bibliographystyle{IEEEtran}
\bibliography{IEEEabrv,reference}

%

%
%
%




\end{document}